\documentclass[journal]{IEEEtran}

\usepackage{epsf}
\usepackage{epstopdf}
\usepackage{cite}
\usepackage{array}
\usepackage{fixltx2e}
\usepackage{float}
\usepackage[cmex10]{amsmath}
\usepackage{amsmath,amsfonts,amssymb}
\usepackage{graphicx,graphics}
\usepackage{enumerate}
\usepackage{color}
\usepackage{verbatim}
\usepackage{amsthm}
\usepackage{multirow}
\usepackage{siunitx}
\usepackage[table,xcdraw]{xcolor}

\usepackage[english]{babel}

\begin{document}

\title{A Multiband OFDMA Heterogeneous Network for Millimeter Wave 5G  Wireless Applications}

\author{Solmaz~Niknam,~\IEEEmembership{Student Member,~IEEE,}
        Ali~A.~Nasir,~\IEEEmembership{Member,~IEEE}
        ~Hani~Mehrpouyan,~\IEEEmembership{Member,~IEEE}
        and Balasubramaniam~Natarajan,~\IEEEmembership{Senior Member,~IEEE.}}

\maketitle

\let\thefootnote\relax\footnote{This work is supported by National Science Foundation (NSF) grant on Enhancing Access to Radio Spectrum (EARS).}
\begin{abstract}
Emerging fifth generation (5G) wireless networks require massive bandwidth in higher frequency bands, extreme network densities, and flexibility of supporting multiple wireless technologies in order to provide higher data rates and seamless coverage. It is expected that utilization of the large bandwidth in the millimeter-wave (mmWave) band and deployment of heterogeneous networks (HetNets) will help address the data rate requirements of 5G networks. However, high pathloss and shadowing in the mmWave frequency band, strong interference in the HetNets due to massive network densification, and coordination of various air interfaces are challenges that must be addressed. In this paper, we consider a relay-based multiband orthogonal frequency division multiple access (OFDMA) HetNet in which mmWave small cells are deployed within the service area of macro cells. Specifically, we attempt to exploit the distinct propagation characteristics of mmWave bands (i.e., $60$ GHz --the \textit{V-band}-- and $70$--$80$ GHz --the \textit{E-band}--) and the Long Term Evolution (LTE) band to maximize overall data rate of the network via efficient resource allocation. The problem is solved using a modified dual decomposition approach and then a low complexity greedy solution based on iterative activity selection algorithm is presented. Simulation results show that the proposed approach outperforms conventional schemes.
\end{abstract}

\section{Introduction} \label{introduction}
The ubiquity of cloud-based applications, ultra-high resolution video streaming, entertainment and many new emerging applications have created an increasing demand for higher data rate in wireless cellular networks. Fifth generation (5G) wireless networks are a solution to that demand because of their promising ability to handle sheer amount of data.
In order to satisfy the data rate demand, the large available bandwidth in millimeter-wave (mmWave) bands is a promising candidate\cite{Rappaport2014Millimeter}. Specifically, the 60 GHz unlicensed frequency band is a valuable frequency resource for offloading traffic from licensed bands. Therefore, 5G networks must be designed to utilize the new frequency bands as well as coexisting and integrating with other radio access technologies  \cite{Andrew2014what}.

Extreme densification is another key enabling technology for increasing capacity that can be realized by deploying heterogeneous networks (HetNets). HetNets with low-power, low-complexity base stations (BSs) such as pico and femto, referred to as small cells (SCs), coupled with conventional macro BSs can potentially improve the overall throughput of cellular networks \cite{andrews2013seven}. Based on Nokia estimates, integrating mmWave SC with 2 GHz carrier bandwidth with macro BS that utilizes current technology such as Long Term Evolution (LTE)-advanced not only can provide a 2 ($Tb/s/km^2$) area capacity, but can also lower the required density for SC deployment \cite{Nokia2015ten}.

One way to enhance system throughput is to incorporate SCs that operate in licensed and unlicensed mmWave band in the coverage area of conventional macrocells. Base stations of SCs can perform a switching procedure between licensed and unlicensed mmWave bands using the transceiver structure in \cite{Hani2014hybrid}. Such a utilization of cells that support licensed and unlicensed mmWave bands and LTE frequencies allow maximum throughput through switching and aggregation and alleviate traffic congestion via offloading.

Advantageous utilization of the two aforementioned enabling technologies necessitates a proper model of network elements and radio resource incorporation within the advanced architecture of HetNets. Specifically, development of a scheme that efficiently utilizes radio resource, including bandwidth and transmission power, and allows switching between different technologies while mitigating intercell and interuser interference is a critical issue \cite{Nokia2015ten}.
Hence, many works have been carried out in this regard. However, the problem of considering frequency bands with \emph{completely different propagation} characteristics to enhance throughput and reduce interference, a huge limiting factor in current wireless cellular networks, is a novel concept. This is especially the case for mmWave frequencies, where the V-band has completely different propagation characteristics compared to E-band and they both differ greatly with respect to lower LTE frequency band (see Fig.~\ref{figure.fig lossVSfreq3}). In this paper, we demonstrate this concept in the context of relay-based multiband orthogonal frequency division multiple access (OFDMA) network. Prior efforts related to resource allocation in OFDMA HetNets have mainly focused on single frequency band operations. This limits their applicability to the proposed model as detailed in the next subsection.

\subsection{Related works} \label{subsec:Related works}
A resource allocation scheme for heterogeneous OFDMA systems with relays has been proposed in~\cite{Xiao2011resource}. However, the problems of relay selection, subcarrier and power allocation are solved separately through disjoint steps which may not be a suitable solution; since, subcarrier allocation is greatly affected by power allocation. \cite{Fooladivanda2013joint} considers a joint resource allocation for HetNets but the total power budget is equally shared among all allocated channels. This can be a limitation for mmWave systems because channel conditions for users can differ considerably as a result of significant human shadowing and pathloss~\cite{Rappaport-14-B}. A resource allocation scheme for OFDMA HetNets is presented in~\cite{Abdelnasser2013subchannel} to maximize network throughput. However, the users are assigned to specific cells without freedom to select the appropriate serving cell. \cite{liang2012game} investigates a game-theoretic resource allocation hierarchy for heterogeneous relay networks, but relay selection is not considered even though it can potentially provide additional degrees of freedom to overcome pathloss and shadowing. In addition to the above approaches and while not considering HetNets, many research articles have considered the important scenario of joint resource allocation in OFDM relay systems~\cite{hsu2011joint, hasan2011resource, sidhu2010joint} but these approaches do not take into account propagation characteristics at mmWave frequencies and are designed to utilize only a single frequency band for communication. ~\cite{hsu2011joint} considers a relay system consisting of one source, one relay, and one destination using a single microwave frequency band, which does not take the benefits of relays and band selection into account. In addition, a minimum rate requirement for each user is not taken into account. Therefore, the work in~\cite{hsu2011joint} may not be well-suited for applications in mmWave systems that must overcome significant propagation issues. A downlink point-to-multipoint wireless network in which the network consists of one source or BS, one relay, and multiple users has been taken into account in~\cite{hasan2011resource}. However, since pathloss and shadowing significantly affect mmWave networks, a single relay may not be able to support multiple users. In addition, the algorithm proposed in~\cite{hasan2011resource} only supports a single band, does not consider subcarrier pairing, and sets no constraints on ensuring a minimum rate for a given user. These omissions may result in loss of communication links, inefficient use of bandwidth, and poor quality of service. Reference~\cite{sidhu2010joint} considers an uplink multi-user transmission scenario in which multiple mobile users communicate with a BS through various relay stations. However in~\cite{sidhu2010joint}, the direct link, or the valuable communication link in the mmWave band is not considered and the authors adopt a high signal-to-noise ratio (SNR) regime to convert the resource allocation problem into a concave optimization problem. The resulting solution is only precise at high SNRs, which may be rarely the case in mmWave systems that suffer from significant pathloss and shadowing. It is worth reiterating that none of the mentioned works utilize multiple frequency bands with very different propagation characteristics and they are all designed to operate in a single frequency band.
\begin{figure}[t]
\centering
    \includegraphics[scale=0.4]{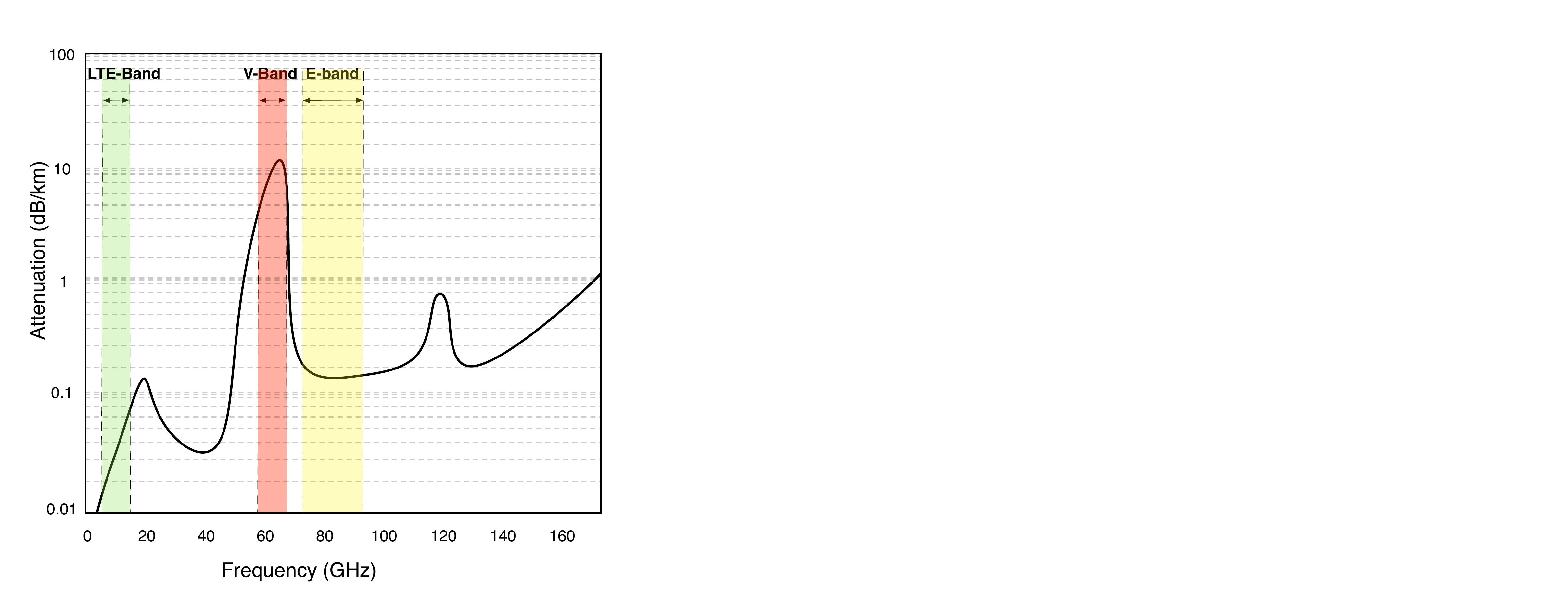}
    \vspace{-5mm}
  \caption{Atmospheric attenuation versus operating frequency~\cite{Rappaport-14-B}.} \label{figure.fig lossVSfreq3}
 \vspace{-1mm}
\end{figure}
\vspace{-2mm}
\subsection{Contributions} \label{subsec:Contributions}
In this paper, we propose a new mmWave HetNet model that uses the mmWave and LTE bands to maximize the overall throughput of the network while meeting power constraints and quality of service (QoS) of each user. Incorporating the LTE band in the model is especially important since mmWave signals are significantly attenuated by atmospheric absorbtion, as shown in Fig.~\ref{figure.fig lossVSfreq3}. In this model, each macro cell contains a macro BS that operates in mmWave and LTE frequency bands and several small BS that operate only in mmWave frequency band serving the outdoor and indoor users, respectively (see Fig.~\ref{figure.fig HetNet}). Confining the small BSs to solely operate in mmWave frequency bands helps reduce interference without sophisticated intercell and interuser interference cancellation techniques. In fact, the strong attenuation and penetration loss of mmWave signals decreases the probability that an outdoor user will be covered by a small BS~\cite{ heath2015coverage }. This phenomenon also results in negligible interference between a macro BS and small cell users and vice versa. Therefore, this model allows the use of the frequencies in the LTE band with favorable propagation characteristics to meet the users' QoS requirements, while the frequencies in the mmWave band are used to reduce the overall level of interference in the network and achieve higher throughputs. The other important advantages of the proposed model can be summarized as follows:
\begin{itemize}
\item By allowing users to adaptively switch between different frequency bands, the model enables the combination of dual air interfaces (i.e., small BSs with licensed and unlicensed mmWave air interface are coupled with macro BS utilizing mmWave and LTE air interface), leading to increased throughput.

\item To support dense networks in a reliable fashion, 5G wireless networks will more extensively use relays~\cite{Osseiran2014Scenarios}. Therefore, we incorporate relays into our model and consider a two-hop (from the BS to the relays and from relays to the users) transmission scheme. Moreover, relays can overcome the severe pathloss and shadowing issue in the mmWave band, especially at 60 GHz, as shown in Fig.~\ref{figure.fig lossVSfreq3}.

\item A new degree of freedom in which a subcarrier applied in the second hop may differ from the one in the first hop, is also considered. This technique, called subcarrier pairing, improves the system sum-rate~\cite{herdin2006chunk} and can be combined with frequency band selection to enhance coverage, especially at mmWave frequencies. For example, if a relay-to-user link distance is longer than the BS-to-relay link and cannot be adequately supported by the V-band (see Fig.~\ref{figure.fig lossVSfreq3}), subcarrier pairing allows for the use of an E-band subcarrier (in the case of a macro-cell base station, the proposed approach can select subcarriers in the E-band or LTE frequency band) with more suitable propagation characteristics to support this link. Our results show that this approach increases the overall sum rate of the network.
\end{itemize}

In this new relay-based multiband OFDMA model, we analytically formulate the optimal resource allocation strategy that maximizes the overall throughput subject to a power constraint and users' QoS requirements. The resulting optimization problem is solved using the dual decomposition approach. Finally, by using the sub-gradient method and an iterative algorithm, the joint resource allocation problem is solved and a suboptimal low complexity greedy solution is presented, as well. Simulations show that the use of smaller cells, relays, and distinct \textit{propagation} characteristics at LTE, V-band and E-band allows the proposed model to overcome large pathloss and shadowing at mmWave frequencies and achieve significantly high data rates.
\begin{figure}[t]
\centering
\includegraphics[scale=0.45]{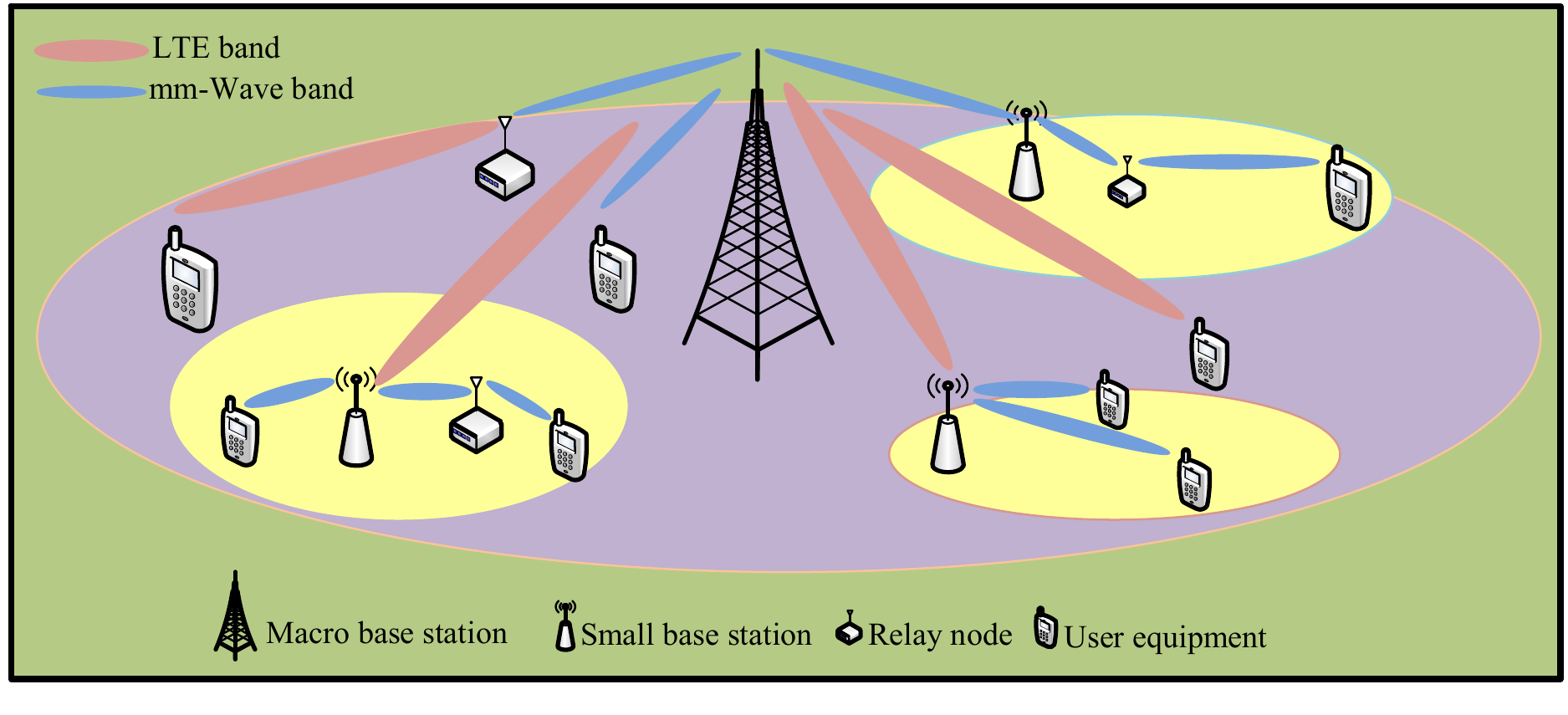}
\caption{Single cell of the HetNet.} \label{figure.fig HetNet}
\end{figure}

\subsection{Organization}
The remainder of this paper is organized as follows. Section~\ref{Proposed Scheme} describes the proposed relay-based OFDMA model and its operation; Section~\ref{Resource Allocation for Hybrid HetNet} formulates the proposed joint resource allocation problem and presents the modified dual decomposition based solution. Sections~\ref{Simulation Result} and~\ref{Conclusion} present simulation results and the conclusion, respectively.\\

\section{Proposed Model} \label{Proposed Scheme}
Comprehensive deployment of 5G networks operating in the V-band face substantial obstacles including high pathloss and significant signal attenuation due to shadowing~\cite{Rappaport2014Millimeter}. Unlike the V-band that suffers from strong gaseous attenuation, the E-band has a large available spectrum in order to support transmissions over longer distances. Our model is an OFDMA HetNet that operates in LTE, V-band and E-band, as shown in Fig.~\ref{figure.fig HetNet}. It consists of several service areas, including a macro BS coupled with $\mathcal{L}$ small BSs $\{{f_0},{f_1},...,{f_l},...,{f_\mathcal{L}}\}$, where $f_0$ is the macro BS that is connected to the small BSs through backhaul links. The BSs mainly attempt to communicate over V-band frequencies but to maintain an acceptable quality of service at low SNR, E-band frequencies can be also used. The macro BS employs the LTE frequency band in addition to mmWave frequency bands in order to serve distant outdoor users for whom mmWave signaling is inefficient due to its strong attenuation. Significant penetration loss at mmWave frequencies ensures that the interference introduced from macro BS to indoor users and from small BS to outdoor users is negligible. In addition, strong channel attenuation and deployment of a highly directional antenna in mmWave systems diminish the effect of interference in the HetNet.

The proposed model assumes that several relay terminals assist in carrying information from BS to the users in order to overcome shadowing by establishing a long-distance connection. Based on channel conditions at different frequency bands, there are two modes of operation: \textit{relay link mode} and \textit{direct link mode}. The BS decides to either utilize the relays or transmit the information through the direct link. If the BS-to-relay and relay-to-user links are favorable compared to the BS-to-user link, then the relay link is used. Otherwise, BS-to-user link is employed. In the first hop, the BS broadcast data to the users and relays over a given subcarrier. Then, in the relay link mode, the relays decode and forward the received data to the users  in the second hop over a subcarrier, which may differ from the subcarrier in the first hop. This subcarrier pairing is performed based on the sum-weighted rate maximization criterion. In fact, the subcarriers in the first and second hops can be in the LTE, E-band and V-band, depending on channel conditions. If the direct link mode is selected, the relays do not forward the received information to the users, and the BSs transmit the data solely through the direct link in the first hop and keep silent in the second hop. A particular relay node can serve more than one user, and a user can receive data from different relays.

Similar to~\cite{hsu2011joint,hasan2011resource,sidhu2010joint}, we assume the network operates in time-division duplex mode, so channel reciprocity is used to estimate all channel parameters at the macro BS. The macro BS is also assumed to perform the resource allocation centrally and provide information to the small BSs using backhaul links that interconnect them to the entire network. This allows simple hardware structures at the small BSs, thereby reducing their cost of deployment. Subsequently, users and relays are informed of the resource allocation parameters during the signaling process that precedes data transmission (i.e., the hand-shaking stage). Due to the fact that subcarriers may carry different types of services, various weights are applied to individual users in the system, determined based on requested service priority, to meet the QoS requirements for each user. The respective resource allocation scheme for this model is provided in the following section.

\vspace{+2mm}
\section{Resource Allocation in the Proposed Model} \label{Resource Allocation for Hybrid HetNet}
In cell $l$ for a user $k_l$ that receives data through the BS and a relay $m_l$ in the first and second hops over subcarrier pair $(i_b,j_{b'})$, the weighted normalized rate can be expressed as \cite{wang2007power}
\begin{equation} \label{equation.equ weighted-normalized-rate}
\mathbb{R}_l = \sum\limits_{\mathcal{S}_l} {\zeta _{{m_l},{k_l}}^{{i_b},{j_{b'}}}\frac{{{w_{{k_l}}}}}{2}{\log_2} \left( {1 + \alpha_{{m_l},{k_l}}^{{i_b},{j_{b'}}}p_{{m_l},{k_l}}^{{i_b},{j_{b'}}}} \right)}, \text{(bits/s/Hz)}
\end{equation}
where the notations are given in Table \ref{tab:Notation}. Additionally,
\begin{itemize}
\item $p_{{m_l},{k_l}}^{{i_b},{j_{b'}}}= p_{{m_l},{k_l}}^{{f_l}\,,({i_b},{j_{b'}})} + p_{{m_l},{k_l}}^{{R_l},({i_b},{j_{b'}})}$ is the aggregate power of user $k_l$ and relay $m_l$ over subcarrier pair (\emph{$i_b,j_{b'}$}) in the first and second hops,
\item $w_{k_l}$ represents the weight of QoS requirements for user \emph{$k_l$},
\item $p_{{m_l},{k_l}}^{{f_l}\,,({i_b},{j_{b'}})}$ denotes the transmit power of the BS $f_l$ to a given user-relay pair $(k_l,m_l)$ over subcarrier \emph{$i_b$} paired with subcarrier \emph{$j_{b'}$} in the second hop,
\item $p_{{m_l},{k_l}}^{{R_l},({i_b},{j_{b'}})}$ shows the transmit power of relay \emph{$m_l$} to user $k_l$, over subcarrier \emph{$j_{b'}$} paired with subcarrier \emph{$i_b$} in the first hop, and
\item $\alpha_{{m_l},{k_l}}^{{i_b},{j_{b'}}}$ is the equivalent channel gain for a given subcarrier pair (\emph{$i_b,j_{b'}$}) allocated to user-relay pair $(k_l,m_l)$, which is determined as
    \begin{equation} \label{equation.equ equivalent alpha}
\hspace{-2mm} \alpha^{i_b,j_{b'}}_{k_l,m_l} \triangleq \left\{
\begin{array}{l l}
\frac{{\alpha _{{f_l}{m_l}}^{{i_b}}\alpha _{{f_l}{k_l}}^{{j_{b'}}}}}{{\alpha _{{f_l}{m_l}}^{{i_b}} + \alpha _{{m_l}{k_l}}^{{j_{b'}}} - \alpha _{{f_l}{k_l}}^{{i_b}}}}, & \;  \text{relay link mode}  \\
\alpha _{{f_l}{k_l}}^{{i_b}}& \;  \text{direct link mode.}
\end{array} \right.
\end{equation}
\end{itemize}
In the direct link mode, the BS and relay powers are given by $p_{{m_l},{k_l}}^{{f_l}\,,({i_b},{j_{b'}})}  \triangleq p_{{m_l},{k_l}}^{{i_b},{j_{b'}}}$ and $p_{{m_l},{k_l}}^{{R_l},({i_b},{j_{b'}})} \triangleq 0$, respectively. $\zeta _{{m_l},{k_l}}^{{i_b},{j_{b'}}},\forall (i_b,j_{b'},k_l,m_l)$ is defined as a binary factor in order to assist mathematical discussion of the resource allocation problem. $\zeta _{{m_l},{k_l}}^{{i_b},{j_{b'}}}=1$ demonstrates that subcarrier pair $(i_b,j_{b'})$ is allocated to the user-relay pair $(k_l,m_l)$. 
\begin{table}[t]
\caption {Notations} \label{tab:Notation}
\centering
\begin{tabular}{|c||p{5cm}|}
\hline
$\mathcal{L}=\{1,...,L\}$         & Set of cells    \\ \hline
$\mathcal{B}=\{1,...,B\}$         & Set of frequency bands   \\ \hline
$\mathcal{N}=\{1,...,N\}$         & Set of subcarriers in each frequency bands    \\ \hline
$\mathcal{M}_l=\{1,...,M_l\}$       & Set of relays in cell $l$   \\ \hline
$\mathcal{K}_l=\{1,...,K_l\}$       & Set of users in cell $l$    \\ \hline
$i_b$                 & Subcarrier in the first hop over frequency bands ${b}$   \\ \hline
$j_{b'}$              & Subcarrier in the second hop over frequency bands ${b'}$   \\ \hline
$m_l$                 & Relay $m$ in cell $l$    \\ \hline
$k_l$                 & User $k$ in cell $l$    \\ \hline
$\alpha^{a}_{cc'}$    & Fading gain of channel between transmitter $c$ and receiver $c'$ over subcarrier $a$ \\ \hline
$\mathcal{S}_l$       & $\mathcal{S}_l=\{(b,b',i,j,m_l,k_l)\mid l\in \mathcal{L}; b,b'\in\mathcal{B}; i,j\in\mathcal{N}; m_l\in\mathcal{M}_l; k_l\in\mathcal{K}_l \}$    \\ \hline
\end{tabular}
\end{table}
\vspace{-3mm}
\subsection{Problem Formulation} \label{sec:Problem Formulation}
The resource allocation problem under a minimum rate requirement for each user and a total power constraint can be formalized as

\begin{subequations}
\label{eq:optimization}
\begin{align}
\label{equation.equ Objective-Function}
\mathbb{P}1: \quad \max_{\boldsymbol{(p,\zeta)}}\sum\limits_l \sum\limits_{\mathcal{S}_l} {\zeta _{{m_l},{k_l}}^{{i_b},{j_{b'}}}\frac{{{w_{{k_l}}}}}{2}{\log_2} \left( {1 + \alpha _{{m_l},{k_l}}^{{i_b},{j_{b'}}}p_{{m_l},{k_l}}^{{i_b},{j_{b'}}}} \right)} \,
\end{align}

subject to

\begin{align}\label{equation.equ first_const}
& {\sum\limits_{j,b,b',m_l,k_l} {\zeta _{{m_l},{k_l}}^{{i_b},{j_{b'}}}}} = 1, \quad \forall i,l\\\label{equation.equ second_const}
&{\sum\limits_{i,b,b',m_l,k_l} {\zeta _{{m_l},{k_l}}^{{i_b},{j_{b'}}}}} = 1, \quad \forall j,l\\
\label{equation.equ third_const}
&\sum\limits_{b,b',i,j,{m_l}} {\frac{{{w_{k_l}}}}{2}\zeta _{{m_l},{k_l}}^{{i_b},{j_{b'}}}{{\log }_2}(1 + \alpha _{{m_l},{k_l}}^{{i_b},{j_{b'}}}p_{{m_l},{k_l}}^{{i_b},{j_{b'}}})}  \ge {R_{\min }}\,\,\,\forall {k_l},l\\
\label{equation.equ forth_const}
&\sum\limits_{\mathcal{S}_l} {\zeta _{{m_l},{k_l}}^{{i_b},{j_{b'}}}}{p_{{m_l},{k_l}}^{{i_b},{j_{b'}}}}  \le {P_l}\,\,\,\,\,\,\forall l\\
\label{equation.equ fifth_const}
&p_{{m_l},{k_l}}^{{i_b},{j_{b'}}} \ge 0,\\ \label{equation.equ sixth_const}
&\zeta _{{m_l},{k_l}}^{{i_b},{j_{b'}}} \in \left\{ {0,1} \right\}\,.
\end{align}
\end{subequations}
${\sum}_{a,b,\ldots,z}$ is used to denote ${\sum}_a{\sum}_b\ldots {\sum}_z$ for simpler representation of the mathematical notations. For each $l$, $\boldsymbol{p}$ and $\boldsymbol{\zeta}$ denote the sets of nonnegative real numbers ${p_{{m_l},{k_l}}^{{i_b},{j_{b'}}}}$ and $\zeta _{{m_l},{k_l}}^{{i_b},{j_{b'}}}$, respectively. Constraints \eqref{equation.equ first_const} and \eqref{equation.equ second_const} correspond to constraints associated with \textit{exclusive} pairing of the subcarriers in the first and second hops. In other words, only one unique subcarrier $i_b$ in the first time hop is paired with subcarrier $j_{b'}$ in the second hop. Furthermore, \eqref{equation.equ forth_const} represents the total power constraint for each BS. Practically, each user in the network also has a minimum rate requirement. We consider \eqref{equation.equ third_const} as a constraint in the resource allocation problem in order to provide each user a minimum rate of $R_{min}$. Aggregate power is obtained by solving the optimization problem and then the power of each small BSs, macro BS, and relay nodes for a specific user over the corresponding subcarrier is calculated using the following equations \cite{wang2007power}

\begin{align}
\label{equation.equ Source_power}
&p_{{m_l},{k_l}}^{{f_l}\,,({i_b},{j_{b'}})} = \frac{{\alpha _{{m_l}{k_l}}^{{j_{b'}}}}}{{\alpha _{{f_l}{m_l}}^{{i_b}} + \alpha _{{m_l}{k_l}}^{{j_{b'}}} - \alpha _{{f_l}{k_l}}^{{i_b}}}}p_{{m_l},{k_l}}^{{i_b},{j_{b'}}},\\
\label{equation.equ Relay_power}
&p_{{m_l},{k_l}}^{{R_l},({i_b},{j_{b'}})} = \frac{{\alpha _{{f_l}{m_l}}^{{i_b}} - \alpha _{{f_l}{k_l}}^{{j_{b'}}}}}{{\alpha _{{f_l}{m_l}}^{{i_b}} + \alpha _{{m_l}{k_l}}^{{j_{b'}}} - \alpha _{{f_l}{k_l}}^{{i_b}}}}p_{{m_l},{k_l}}^{{i_b},{j_{b'}}}.
\end{align}
Since the optimization problem in~\eqref{eq:optimization} consists of a binary parameter $\zeta _{{m_l},{k_l}}^{{i_b},{j_{b'}}}$, solving it requires application of integer programming, which has excessive computational complexity \cite{boyd2009convex}. In order to make the problem tractable, we relax the integer factor such that it can be real values equal or greater than zero. After relaxing the constraint in~\eqref{equation.equ sixth_const}, the optimization problem in~\eqref{eq:optimization} can be rewritten as

\begin{align}\label{equation.equ relaxed objective function}
\mathbb{P}2: \quad &\max_{\boldsymbol{(p,\zeta)}}\sum\limits_l \sum\limits_{\mathcal{S}_l} {\zeta _{{m_l},{k_l}}^{{i_b},{j_{b'}}}\frac{{{w_{{k_l}}}}}{2}{\log_2} \left( {1 + \frac{{\alpha _{{m_l},{k_l}}^{{i_b},{j_{b'}}}p_{{m_l},{k_l}}^{{i_b},{j_{b'}}}}}{{\zeta _{{m_l},{k_l}}^{{i_b},{j_{b'}}}}}}\right)}\\ \notag
 &\hspace{1mm}\text{s.t.} \qquad \eqref{equation.equ first_const}, \eqref{equation.equ second_const}, \eqref{equation.equ forth_const}, \eqref{equation.equ fifth_const},\\
\label{equation.equ relaxed zeta&min_r}
\sum\limits_{b,b',i,j,{m_l}} &{\frac{{{w_{k_l}}}}{2}\zeta _{{m_l},{k_l}}^{{i_b},{j_{b'}}}{\log_2}\Bigg( {1 + \frac{{\alpha _{{m_l},{k_l}}^{{i_b},{j_{b'}}}p_{{m_l},{k_l}}^{{i_b},{j_{b'}}}}}{{\zeta _{{m_l},{k_l}}^{{i_b},{j_{b'}}}}}}\Bigg)}  \ge {R_{\min }}\,\,\,\forall {k_l},l, \\
&\text{and} \quad \zeta _{{m_l},{k_l}}^{{i_b},{j_{b'}}} \ge 0,\,\forall (i_b,j_{b'},m_l,k_l).
\end{align}

\newtheorem{lemma}{Lemma}
\begin{lemma}\label{lem:lemma1}
The objective function in $\mathbb{P}2$ is concave in $p$ and $\zeta$.\\
\end{lemma}
\begin{proof}
Let $q(x){=}x\log\left( {1 + y/x} \right)\left| {_{y = ax + b}} \right.{=}x\log \left( {1 + a + \frac{b}{x}} \right)$. Then, for $x>0$, $q(x)$ is a concave function (This can be demonstrated by taking its second derivative). Consequently, $x\log\left( {1 + y/x} \right)$ is concave, since its restriction to any line, i.e. $q(x)$, is concave~\cite{boyd2009convex}. Therefore, the objective function in $\mathbb{P}2$ is a concave function in that it is a nonnegative weighted sum of concave functions in the form of $x\log(1+y/x)$.
\vspace{-7mm}
\end{proof}

Since $\mathbb{P}2$, is a convex optimization problem, it can be solved by any standard method of solving convex problems. However, the value of $\zeta _{{m_l},{k_l}}^{{i_b},{j_{b'}}}$ may not be integer. Therefore, we proceed with the modified two-stage dual decomposition method as we discuss in the following. It is worth mentioning that if the number of subcarriers is adequately large, then the \textit{duality gap} of a non-convex optimization problem reduces to zero \cite{yu2006dual}.
The dual problem is given by
\begin{align} \label{equation.equ Dual-Problem}
\mathbb{P}3: \quad &\mathop {\min\limits_{\boldsymbol{\tau,\delta}}} D\left(\boldsymbol{\delta,\tau} \right) = \mathop {\min\limits_{\boldsymbol{\tau,\delta}}}\mathop {\max\limits_{\boldsymbol{p,\zeta}}} \,\,L\left( \boldsymbol{p,\zeta,\delta,\tau} \right)\\ \notag
 &\hspace{1mm}\text{s.t.} \quad  \eqref{equation.equ first_const} \;\;\;
\text{and}
\;\;\;
\eqref{equation.equ second_const},
\end{align}
where the Lagrangian is given by
\begin{align}\label{equation.equ Lagrangian} \notag
\hspace{-3mm} L(\boldsymbol{p,\zeta ,\tau ,\delta} )&= \sum\limits_l \sum\limits_{\mathcal{S}_l} {\zeta _{{m_l},{k_l}}^{{i_b},{j_{b'}}}\frac{{{w_{{k_l}}}}}{2}{\log_2} \Big( {1 + \frac{{\alpha _{{m_l},{k_l}}^{{i_b},{j_{b'}}}p_{{m_l},{k_l}}^{{i_b},{j_{b'}}}}}{{\zeta _{{m_l},{k_l}}^{{i_b},{j_{b'}}}}}}\Big)}
\\ \notag
&+\sum\limits_l {{\tau _l}\Big( {{P_l} - \sum\limits_{\mathcal{S}_l}{p_{{m_l},{k_l}}^{{i_b},{j_{b'}}}}}\Big)}
\\ \notag
&+\sum\limits_{l ,k_l} {\delta _l^{{k_l}}\bigg( \sum\limits_{b,b',{i_b},{j_{b'}},m} {\frac{w_k}{2}\zeta _{{m_l},{k_l}}^{{i_b},{j_{b'}}}}}\\
&\hspace{+0.5cm}\times {{{\log }_2}\Big(1 + \frac{{\alpha _{{m_l},{k_l}}^{{i_b},{j_{b'}}}p_{{m_l},{k_l}}^{{i_b},{j_{b'}}}}}{{\zeta _{{m_l},{k_l}}^{{i_b},{j_{b'}}}}}\Big)}  - {R_{\min }}\bigg).
\end{align}
In \eqref{equation.equ Lagrangian}, $\delta _l^{k_l}$ and $\tau_l$ are Lagrangian multipliers.
\begin{lemma}\label{lem:lemma2}
For a given $\zeta _{{m_l},{k_l}}^{{i_b},{j_{b'}}}$, the optimal power allocation that maximizes $L(\boldsymbol{p,\zeta ,\tau ,\delta} )$ is given by

\begin{equation} \label{equation.equ OptPower}
{\left( {p_{{m_l},{k_l}}^{{i_b},{j_{b'}}}} \right)^ * } = \zeta _{{m_l},{k_l}}^{{i_b},{j_{b'}}}\underbrace{{\left[ {\frac{{\left( {1 + \delta _l^{{k_l}}} \right){w_{k_l}}}}{{2{\tau _l}}} - \frac{1}{{\alpha _{{m_l},{k_l}}^{{i_b},{j_{b'}}}}}} \right]^+}}_{g_{{m_l},{k_l}}^{{i_b},{j_{b'}}}},
\end{equation}
where $[x]^+$  indicates $\max(0,x)$.\\
\end{lemma}
\begin{proof}
By applying \emph{Karush-Kuhn-Tucker} (KKT) condition to the Lagrangian function, we have

\begin{align}\label{KKT} \notag
&\frac{{\partial L}}{{\partial p_{{m_l},{k_l}}^{{i_b},{j_{b'}}}}} = \sum\limits_l \sum\limits_{{S_l}} \Bigg( {\left( {1 + \delta _l^{{k_l}}} \right)\zeta _{{m_l},{k_l}}^{{i_b},{j_{b'}}}\frac{{{w_{{k_l}}}}}{2}}\\ \notag
&\hspace{2.5cm}\times \frac{{\frac{{\alpha _{{m_l},{k_l}}^{{i_b},{j_{b'}}}}}{{\zeta _{{m_l},{k_l}}^{{i_b},{j_{b'}}}}}}}{{1 + \frac{{\alpha _{{m_l},{k_l}}^{{i_b},{j_{b'}}}p_{{m_l},{k_l}}^{{i_b},{j_{b'}}}}}{{\zeta _{{m_l},{k_l}}^{{i_b},{j_{b'}}}}}}}- {\tau _l} \Bigg)=0\\\notag
&\Longrightarrow\left( {1 + \delta _l^{{k_l}}} \right)\frac{{{w_{{k_l}}}}}{2}\frac{{\alpha _{{m_l},{k_l}}^{{i_b},{j_{b'}}}}}{{1 + \frac{{\alpha _{{m_l},{k_l}}^{{i_b},{j_{b'}}}p_{{m_l},{k_l}}^{{i_b},{j_{b'}}}}}{{\zeta _{{m_l},{k_l}}^{{i_b},{j_{b'}}}}}}} = {\tau _l}\\ \notag
&\Longrightarrow\left( {1 + \delta _l^{{k_l}}} \right)\frac{{{w_{{k_l}}}}}{2}\frac{{\alpha _{{m_l},{k_l}}^{{i_b},{j_{b'}}}}}{{{\tau _l}}} - 1 = \frac{{\alpha _{{m_l},{k_l}}^{{i_b},{j_{b'}}}p_{{m_l},{k_l}}^{{i_b},{j_{b'}}}}}{{\zeta _{{m_l},{k_l}}^{{i_b},{j_{b'}}}}}\\ \notag
&\Longrightarrow {p_{{m_l},{k_l}}^{{i_b},{j_{b'}}}}  = \zeta _{{m_l},{k_l}}^{{i_b},{j_{b'}}}\left({\frac{{\left( {1 + \delta _l^{{k_l}}} \right){w_{k_l}}}}{{2{\tau _l}}} - \frac{1}{{\alpha _{{m_l},{k_l}}^{{i_b},{j_{b'}}}}}} \right)
\end{align}
Considering constraint \eqref{equation.equ fifth_const}, the power values must be positive. Therefore, \eqref{equation.equ OptPower} gives the optimal power expression.
\end{proof}
\begin{table}[t]
\begin{center}
\caption {Resource Allocation Algorithm} \label{tab:Algorithm}
\begin{tabular}{p{7cm}}
\hline
\textbf{\textbf{\emph{Algorithm of the joint resource allocation}}} \\
\hline
\emph{\textbf{1}: Initialize the Lagrangian multipliers (first iteration) and generate the channel fading gains (Alpha parameters)}.\\
\emph{\textbf{2}: Find the $(i_b,j_{b'},m_l,k_l)$ that maximize Z, and set the corresponding} $\zeta^{i_b,j_{b'}}_{m_l,k_l} =1$.\\
\emph{\textbf{3}: Find the optimal value of the power} ${p^\ast} _{{m_l},{k_l}}^{{i_b},{j_{b'}}}$ \emph{by} (\ref{equation.equ OptPower}).\\
\emph{\textbf{4}: Update Lagrangian multipliers.}\\
\emph{\textbf{5}: Iterate the above steps until all Lagrangian multipliers converge. Iteration will stop when}\\
\hspace{1cm}$\frac{{\left| {{\delta ^{(l + 1)}} - {\delta ^{(l)}}} \right|}}{{\left| {{\delta ^{(l + 1)}}} \right|}} < \varepsilon _\delta, \frac{{\left| {{\tau ^{(l + 1)}} - {\tau ^{(l)}}} \right|}}{{\left| {{\tau ^{(l + 1)}}} \right|}} < \varepsilon _\tau.$\\
\emph{\textbf{6}: End}\\
\hline
\end{tabular}
\end{center}
\vspace{-5mm}
\end{table}

\begin{lemma}\label{lem:lemma3}
The integer valued $\zeta _{{m_l},{k_l}}^{{i_b},{j_{b'}}}$ that maximizes $L(\boldsymbol{p,\zeta ,\tau ,\delta} )$ corresponds to

\begin{equation} \label{equation.equ arg-max}
{\zeta^\ast} _{{m_l},{k_l}}^{{i_b},{j_{b'}}} = \left\{
\begin{array}{l l}
1 & \hspace{1mm} \text {\emph{$\,({i_b},{j_{b'}},{m_l},{k_l})$}=$\arg {\rm{ }}\mathop {\max }\limits_{k_l,m_l,i_b,j_{b'}} {\rm{ }}Z_{{m_l},{k_l}}^{{i_b},{j_{b'}}}$} \\
0 & \hspace{2mm} \text{otherwise}.
\end{array} \right.
\end{equation}
\end{lemma}
\begin{proof}
In order to find the optimal value of the binary factor $\zeta _{{m_l},{k_l}}^{{i_b},{j_{b'}}}$, we substitute \eqref{equation.equ OptPower} into \eqref{equation.equ Dual-Problem} and obtain

\begin{align} \label{equation.equ Dual-Function} \notag
D\left(\boldsymbol{\delta,\tau} \right) =& \mathop {\max\limits_{\boldsymbol{\zeta}}} \,\,L\left( \boldsymbol{p^*,\zeta,\delta,\tau} \right)\\
&\hspace{1mm} \text{s.t.} \quad  \eqref{equation.equ first_const} \,\,\, \text{and}\,\,\, \eqref{equation.equ second_const}.
\end{align}
Subsequently, it can be rewritten as

\begin{align} \notag
\quad &\mathop {\max\limits_{\boldsymbol{\zeta}}} L\left( \boldsymbol{p^*,\zeta ,\tau ,\delta } \right) = \mathop {\max\limits_{\boldsymbol{\zeta}}} \bigg(\sum\limits_l \sum\limits_{\mathcal{S}_l} {\zeta _{{m_l},{k_l}}^{{i_b},{j_{b'}}}} Z_{{m_l},{k_l}}^{{i_b},{j_{b'}}} \\
&\hspace{3cm}+ C\Big( {{\tau _l},\delta _l^{{k_l}}} \Big)\bigg),
\end{align}
where

\begin{align} \label{equation.equ Z}
&Z_{{m_l},{k_l}}^{{i_b},{j_{b'}}} = \frac{{\left( {1 + \delta _l^{{k_l}}} \right){w_{k_l}}}}{2}{\log _2}(1 + \alpha _{{m_l},{k_l}}^{{i_b},{j_{b'}}}g_{{m_l},{k_l}}^{{i_b},{j_{b'}}}) - {\tau _l}g_{{m_l},{k_l}}^{{i_b},{j_{b'}}},\\
\label{equation.equ C}
&C\left( {{\tau},{\delta}} \right) = \sum\limits_{l = 0}^L {{\tau _l}{P_l}}  - \sum\limits_{{k_l} = 1}^{{K_l}} {\delta _l^{{k_l}}{R_{\min }}} .
\end{align}
Since, \eqref{equation.equ C} is constant with respect to $\zeta _{{m_l},{k_l}}^{{i_b},{j_{b'}}}$, the maximum value of the Lagrangian is achieved by adopting the maximum values of $Z _{{m_l},{k_l}}^{{i_b},{j_{b'}}}$.
\end{proof}

The Lagrangian multipliers are updated via~\eqref{equation.eq Lagrangian-multiplier-delta} and~\eqref{equation.eq Lagrangian-multiplier-tau} in the next page. A brief algorithm of the joint resource allocation problem is summarized in Table~\ref{tab:Algorithm}.
\begin{figure*}[!h]
\begin{align} \label{equation.eq Lagrangian-multiplier-delta}
{\left( {{\tau _l}} \right)^{n + 1}} =& {\left( {{\tau _l}} \right)^n} - \varepsilon _\tau\left( {{P_l} - \sum\limits_{\mathcal{S}_l}{\zeta _{{m_l},{k_l}}^{{i_b},{j_{b'}}}\hspace{0.1cm}g_{{m_l},{k_l}}^{{i_b},{j_{b'}}}} } \right) \forall l,\\ \label{equation.eq Lagrangian-multiplier-tau}
 {\left( {\delta _l^{{k_l}}} \right)^{n + 1}} =& {\left( {\delta _l^{{k_l}}} \right)^n} - \varepsilon _\delta\left( {\sum\limits_{b,b',{i_b},{j_{b'}},{m_l}} {\frac{{{w_{k_l}}}}{2}\hspace{0.1cm}\zeta _{{m_l},{k_l}}^{{i_b},{j_{b'}}}\hspace{0.1cm}{{\log }_2}(1 + \alpha _{{m_l},{k_l}}^{{i_b},{j_{b'}}}\hspace{0.1cm}g_{{m_l},{k_l}}^{{i_b},{j_{b'}}})}  - {R_{\min }}} \right)\forall {k_l},l.\\
\hline
\notag
\end{align}
\vspace{-1.5cm}
\end{figure*}
\vspace{+3mm}

\subsection{Suboptimal Solution} \label{sec:Sub-optimal Solution}
The solution discussed in the previous subsection is a near-optimal solution; since, the integer factor $\zeta$ is relaxed. As mentioned, the duality gap decreases by increasing the number of subcarriers. However, the computational complexity of the problem increases as the number of subcarriers, relays, users, and the SCs increase, i.e., $O\left(KM(L+1)BN!\right)$. The computational complexity is mainly due to the search for the optimal value of the factor $\zeta$. Therefore, we provide a suboptimal greedy solution to find $\zeta$, based on the ``Greedy Iterative Activity Selection'' algorithm \cite{Thomas2001algorithms}. Generally, given a set of activities, this algorithm enables the selection of a subset of non-conflicting activities to perform within a time frame, which leads to maximizing the objective function. In this case, we treat resources as non-conflicting activities due to the unique allocation property. Subsequently, by sorting their corresponding $Z$ values, we try to add the resource indices with the highest $Z$ values to the subset of selected activities. The steps involved in the greedy algorithm are outlined in Table~\ref{tab:SuboptimalAlgorithm}.
\begin{table}[b]
\centering
\caption{Pathloss Model Parameters}
\label{tab:PathlossParam}
\begin{tabular}{c|c|c|c|c|c|}
\cline{2-6}
                            & \multicolumn{2}{c|}{$V$-band}                                      & \multicolumn{2}{c|}{$E$-band}                                      & LTE                 \\ \cline{2-6}
\multirow{-2}{*}{}          & \cellcolor[HTML]{C0C0C0}indoor & \cellcolor[HTML]{C0C0C0}outdoor & \cellcolor[HTML]{C0C0C0}indoor & \cellcolor[HTML]{C0C0C0}outdoor &                     \\ \cline{1-5}
\multicolumn{1}{|c|}{$\beta$}  & $2.5$                            & $2.2$                             & $2$                              & $2.1$                             & \multirow{-1}{*}{$2$} \\ \hline
\multicolumn{1}{|c|}{$\psi_{dB}$}   & $5.4 dB$                            & $5 dB$                               & $4.7 dB$                            & $2.1 dB$                             & $4 dB$                   \\ \hline
\multicolumn{1}{|l|}{$\gamma$} & \multicolumn{2}{c|}{$9.4$}                                         & \multicolumn{2}{c|}{$2$}                                           & $2$                   \\ \hline
\end{tabular}
\end{table}
\begin{table}[b]
\centering
\caption{Simulation Parameters}
\label{tab:SimulationParameters}
\begin{tabular}{c|c|c|c|c|}
\cline{2-5}
                                             & \multicolumn{2}{c|}{\cellcolor[HTML]{9B9B9B}Small cells}                 & \multicolumn{2}{c|}{\cellcolor[HTML]{9B9B9B}Macro cell}                  \\ \hline
\multicolumn{1}{|c|}{$P_l$}                    & \multicolumn{2}{c|}{$3 dB$}                                                & \multicolumn{2}{c|}{$16 dB$}                                               \\ \hline
\multicolumn{1}{|c|}{$R_{min}$}                 & \multicolumn{2}{c|}{$3 $ bits/s/Hz}                                          & \multicolumn{2}{c|}{$3$ bits/s/Hz}                                          \\ \hline
\multicolumn{1}{|c|}{}                       & \cellcolor[HTML]{C0C0C0}Relay link & \cellcolor[HTML]{C0C0C0}Direct link & \cellcolor[HTML]{C0C0C0}Relay link & \cellcolor[HTML]{C0C0C0}Direct link \\ \cline{2-5}
\multicolumn{1}{|c|}{\multirow{-1}{*}{$d_0(m)$}} & $[10\,\,30]$                        & $[10\,\,50]$                        &$[100\,\,300]$                    & $[50\,\,500] $                       \\ \hline
\multicolumn{1}{|c|}{$\varepsilon_\delta$, $\varepsilon_\tau$}                & \multicolumn{4}{c|}{$10^{-4}$}                                                                                                                              \\ \hline
\multicolumn{1}{|c|}{step size}                & \multicolumn{4}{c|}{$0.5/\sqrt{n}$, $n$ \text{denotes the iteration index}}                                                                                                                              \\ \hline
\end{tabular}
\end{table}
\section{Simulation Results} \label{Simulation Result}
\begin{figure}[t]
\centering
\includegraphics[width=0.45\textwidth]{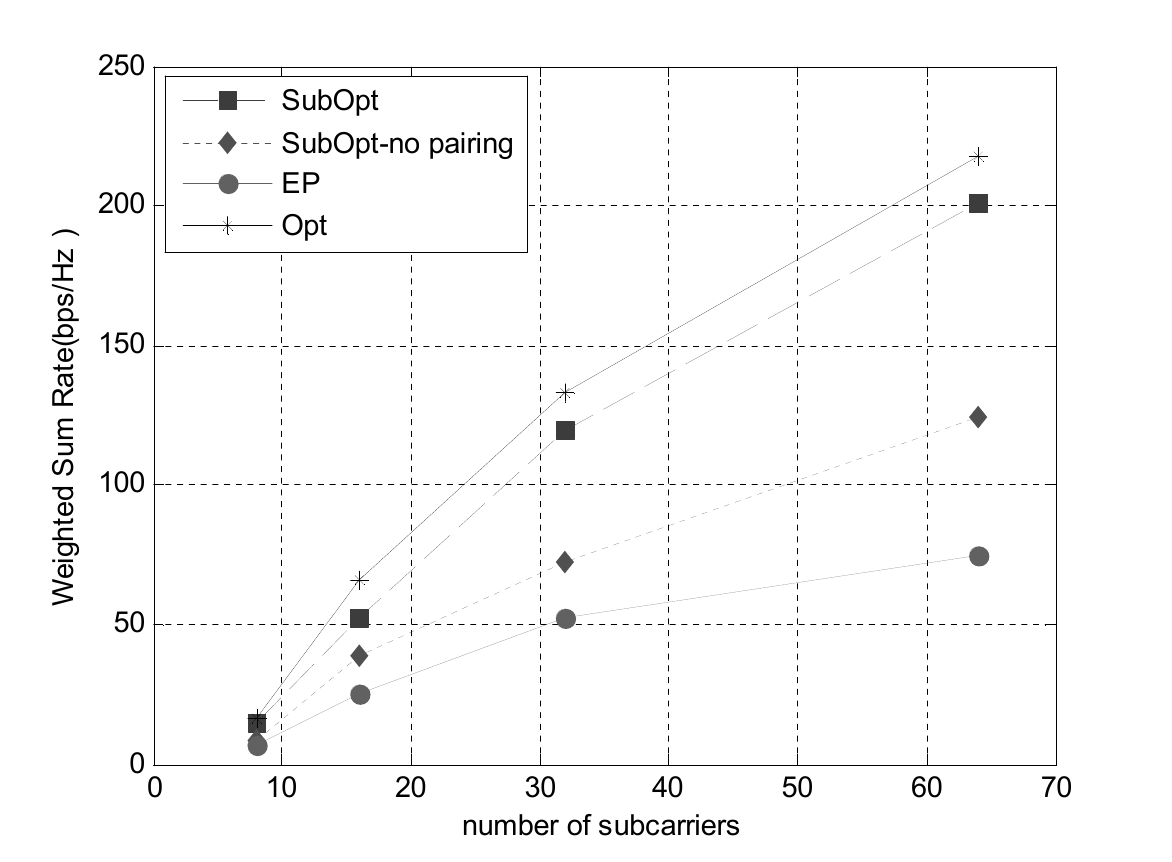}
\caption{Sum-weighted rate versus number of subcarrier for $K_l=4$ and $M_l=2$ in each cell, $L=3$ cell.}\label{SWR_N}
\end{figure}
\begin{figure}[t]
\centering
\includegraphics[width=0.45\textwidth]{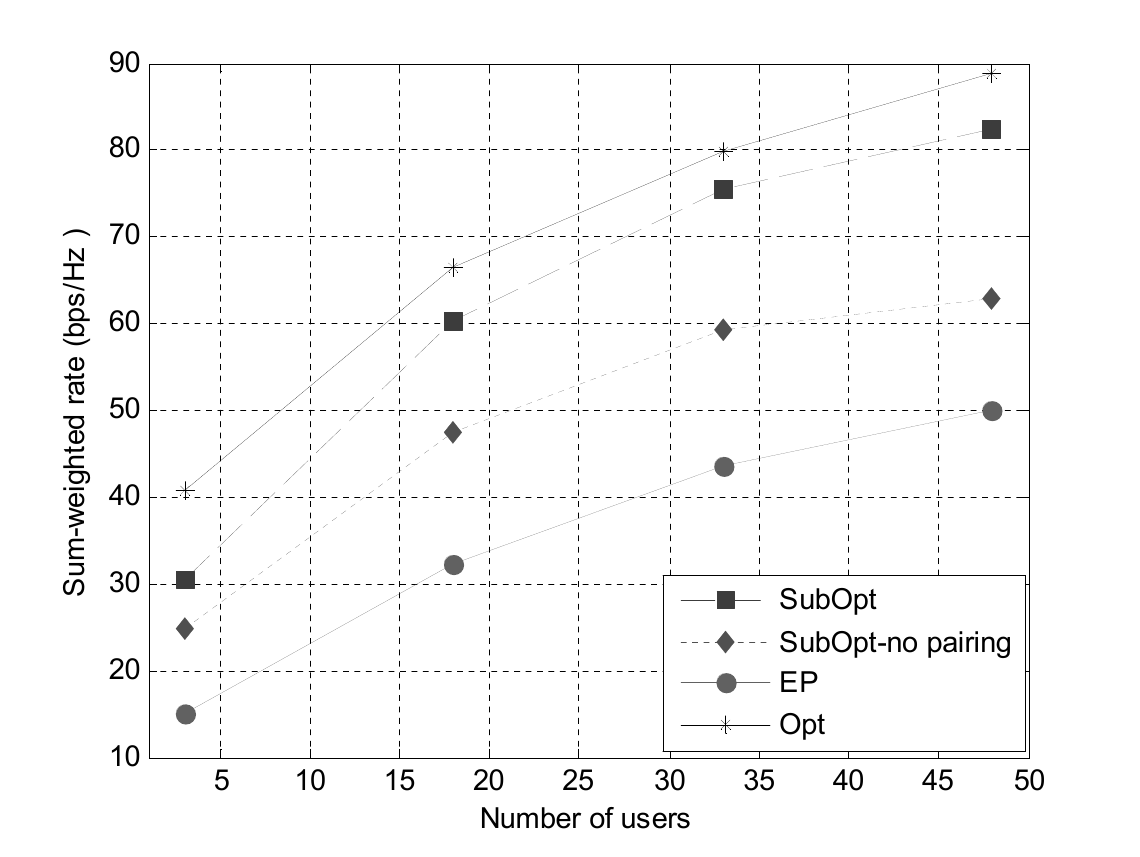}
\caption{Sum-weighted rate versus number of user for $N=15$ in each band, $M_l=5$ in each cell, $L=3$ cell.}\label{SWR_K}
\end{figure}
\begin{figure}[t]
\centering
\includegraphics[width=0.42\textwidth]{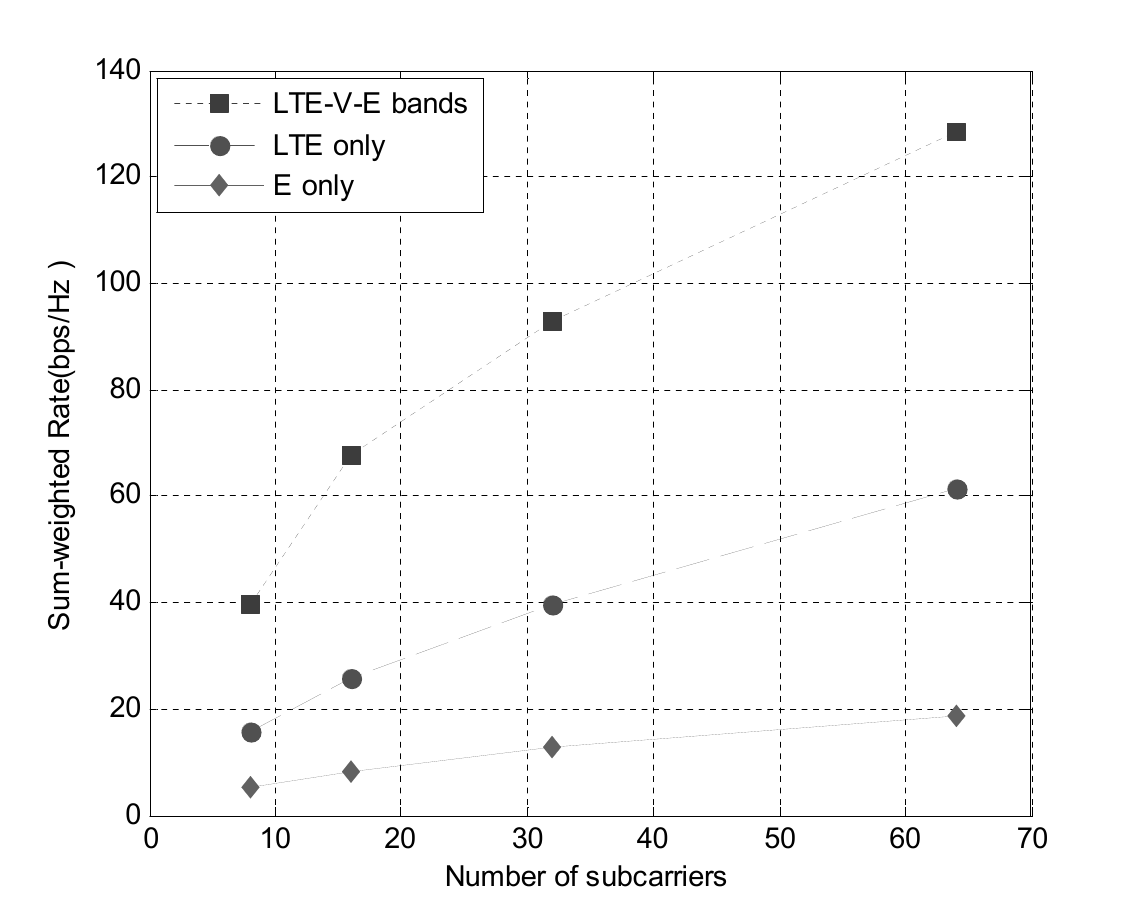}
\vspace{-.15cm}
\caption{Sum-weighted rate versus number of subcarrier for $M_l=2$, $K_l=2$ in each cell, $L=3$ cell.}\label{3V1}
\end{figure}
\begin{figure}[t]
\centering
\includegraphics[width=0.45\textwidth]{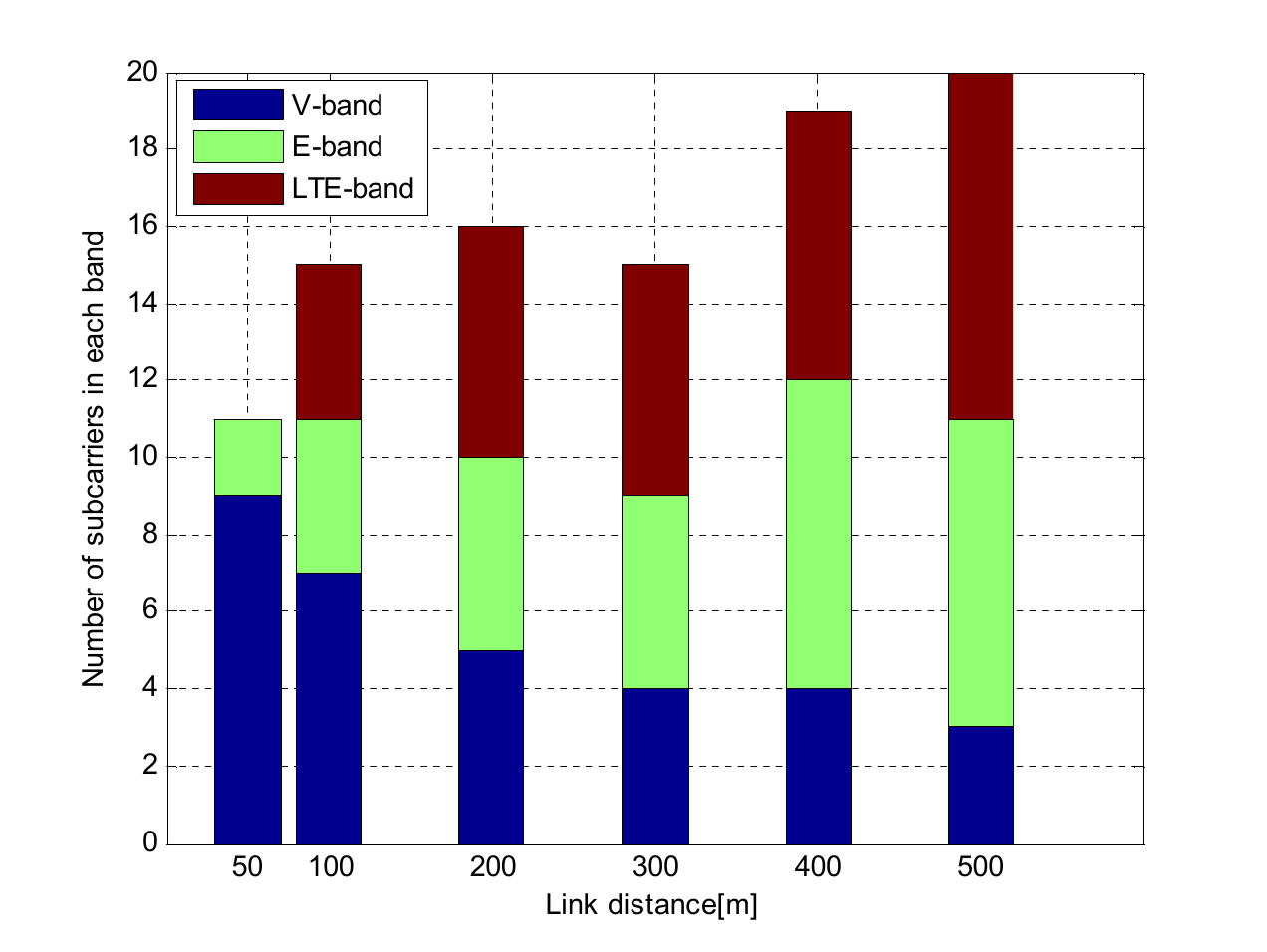}
\caption{Number of subcarriers in each band versus link distance for $M_l=2$, $K_l=2$ and $N=8$ in each cell, $L=3$ cell.}\label{NvsDisstance}
\end{figure}
In this section, we present simulation results that demonstrate the advantage of the proposed scheme in enhancing the overall sum-weighted rate of HetNets by deploying the E-band and V-band SCs and considering their specific propagation characteristics within the service area of macro BSs.
The presence of line of sight (LoS) of links is assumed given the high antenna directivity in the mmWave band. We also consider the extended large-scale path-loss model, which is dependent on the distance and frequency of operation \cite{Rappaport-14-B}, expressed as

\begin{align}
{\rm{PL}}\left( {{\rm{\emph{d}}},{\rm{\emph{f}}}} \right)\left( {{\rm{dB}}} \right){\rm{ = }}\gamma .{\rm{10lo}}{{\rm{g}}_{10}}\left( {\frac{f}{{{f_0}}}} \right){\rm{ + }}\beta .{\rm{10lo}}{{\rm{g}}_{10}}\left( {\frac{d}{{{d_0}}}} \right){\rm{ + }}{\chi _{\text{dB}}},
\end{align}
where ${{f}/{{{f_0}}}}$ and ${{d}/{{{d_0}}}}$ are the ratio of the frequency and distance deviation about the center carrier frequency and reference distance, respectively. $\beta$ and $\gamma$ are path-loss exponent and frequency-dependency factor, respectively. $\chi _{dB}$ is the shadowing factor which is a zero mean Gaussian random variable with standard deviation $\psi_{dB}$. In order to model large-scale fading, the reference distance is set to be $d_0=5$ meter (resp. $d_0=10$ meter) for the small cells (resp. macro cell). The distance between transmitters and receivers, in the small cells (resp. macro cell), are random in interval $[10\,\,50]$ and $[10\,\, 30]$ (resp. $[50\,\,500]$ and $[100\,\,300]$) meters in the direct and relay links, respectively. Although, we use a generic deployment scenario, the results here can be extended to any specific scenario for example, the ones in 3rd Generation Partnership Project (3GPP). We set the other large scale parameter, i.e., pathloss exponent to 2.5 and 2 for the 60 GHz and 70–-80 GHz bands, in the small cells (resp. macro cell), respectively~\cite[p.~106]{Rappaport-14-B}. Moreover, the shadowing effect of the channels is modeled by a zero mean Gaussian random variable with standard deviation $5.4$ dB and $4.7$ dB (resp. $5$ dB and $2.1$ dB) for the 60 GHz and 70--80 GHz bands, respectively~\cite{Rappaport-14-B}. For the LTE band, pathloss exponent and shadowing standard deviation are set to $2$ and $4$ dB, respectively~\cite{Khan2009LTE}.
\begin{table}[t]
\begin{center}
\caption {Suboptimal Solution} \label{tab:SuboptimalAlgorithm}
\begin{tabular}{p{5cm}}
\hline
\textbf{\textbf{\emph{Greedy Algorithm}}} \\
\hline
Sort $\mathcal{S}_l$ by corresponding $Z$ values, for every $l$ \\
$\mathcal{S}^\ast_l={\mathcal{S}_l(1)}$\\
$r=1$\\
$n=length(\mathcal{S}_l)$\\
for $q=2:n$\\
\quad if $Z_{\mathcal{S}^\ast_l}(q)\geq Z_{\mathcal{S}^\ast_l}(r)$\\
\quad $\mathcal{S}^\ast_l=\mathcal{S}^\ast_l \cup {\mathcal{S}^\ast_l}(q) $\\
\quad $r=q$\\
\quad end\\
end\\
\hline
\end{tabular}
\end{center}
\end{table}
In the proposed scheme, initial Lagrangian multipliers are randomly set and the step size for the subgradient method is set to $0.5/\sqrt{n}$, where $n$ denotes the iteration index. $\varepsilon_\delta$ and $\varepsilon_\tau$ are set to be $10^{-4}$. The weights $w_{k_l}$ are considered to be $w_{k_l}=1+(k_l-1)/(K_l-1), \forall k_l,l$ which is used only as an example. The minimum rate requirement for the users is set to 3 (bits/sec/Hz) and the total transmit power is
set to $3$ dB (resp. $16$ dB) for the small BSs (resp. macro BS). For clarity purposes, a list of the parameters set in the simulation are provided in Table \ref{tab:PathlossParam} and \ref{tab:SimulationParameters}. As for the performance comparison, some other comparable and related approaches of resource allocation as well as the optimal solution as the upper bound with maximum throughput and highest complexity are considered:
\begin{itemize}
\item \textbf{\emph{EP}}: The conventional equal power scheme in which power is equally allocated to all subcarriers.
\item  \textbf{\emph{SubOpt no-pairing}}: Scheme in which subcarrier pairing technique is not taken into account.
\item  \textbf{\emph{LTE only}}: Scheme where only LTE frequency band is considered.
\item  \textbf{\emph{E only}}: Scheme where only mmWave E-band is considered.
\item \textbf{\emph{SubOpt}}: The greedy solution presented in this paper.
\end{itemize}
The sum-weighted rate of the network versus the number of subcarriers with, $M_l=2$ relays and $K_l=4$ in each cell is shown in Fig.~\ref{SWR_N}. As anticipated, it can be observed that by increasing the number of subcarriers, the sum-weighted rate of the network increases for all three schemes. However, the proposed joint resource allocation outperforms the other two schemes. Moreover, the largest gain is obtained by applying power allocation; while subcarrier pairing also provides reasonable gains.

Fig.~\ref{SWR_K} shows the sum-weighted rate versus the number of users in the network, with $M_l=5$ relays in each cell and $N=15$ subcarriers in each frequency band. We observe a similar pattern to Fig.~\ref{SWR_N}, where power allocation provides the largest gain, followed by subcarrier pairing. This can primarily be attributed to the nature of mmWave channels, which are significantly affected by shadowing and pathloss. Subcarrier pairing also provides the resource allocation algorithm with flexibility to switch between frequency bands based on channel conditions. This is the main reason that the proposed algorithm outperforms the approach with no subcarrier pairing, as shown in Figs.~\ref{SWR_N} and~\ref{SWR_K}.

Fig.~\ref{3V1} demonstrates the sum-weighted rate of the network in three cases: one case in which the network utilizes the proposed scheme while applying the mentioned resource allocation and two cases in which either LTE band or E-band is used. It shows that the E-band can be effectively used to overcome significant path loss at V-band frequencies. In fact, because of the strong signal attenuation in the V-band, communication over this band is only possible over short-range distances, so, users that are far from BSs and/or relays experience poorer received SNRs, causing a reduction in data rate. By including the E-band in the resource allocation problem, however a HetNet can take advantage of lower pathloss in the E-band in order to enhance the overall sum-weighted rate of the network. Fig.~\ref{3V1} indicates the importance of utilizing various bands within the mmWave band for future HetNets. It can also be seen in Fig.~\ref{NvsDisstance}, that for the longer distances when V-band cannot provide the desired quality of service, the E-band and LTE subcarriers can be used instead. As expected, adding more resources leads to an increase in sum-weighted rate. However, this proposed scheme opportunistically utilizes an unlicensed frequency band that was previously underutilized.

\section{Conclusion and Future Works} \label{Conclusion}
This paper proposed a new scheme for utilizing mmWave frequency bands with distinct propagation characteristics in a HetNet structure. A resource allocation problem that considers utilization of E-band and V-band along with LTE band, subcarrier allocation, subcarrier pairing, and relay selection was formulated. The proposed scheme was applied to the downlink scenario of a HetNet with a macro BS coupled with small mmWave BSs. Our objective was to maximize the sum-weighted rate of the network while considering a minimum rate requirement for each user. The resulting optimization problem was solved by considering its dual form. Subsequently, a suboptimal solution was presented. Simulation results showed that our scheme outperformed conventional schemes and demonstrated that utilization of E- and V-bands can play a major role in addressing the propagation challenges at mmWave frequencies. In future work, it would be interesting to consider base stations equipped with multiple transceiver antennas and evaluate the effect of MIMO technique on the presented model. In addition, considering high pathloss and shadowing effects in mmWave frequency bands and therefore the need for highly directional beam antennas to have an acceptable link quality, obstacles and blockages in the environment have a considerable effect on the strength and quality of the received signal. Since, mmWave signals are greatly narrow and can be easily blocked even by human body-sized obstacles. Therefore, it would be important to consider the impact of the beamwidth of the antennas and the density of the obstacles on the network performance. Finally, backhaul configurations that can support the proposed scheme would also be of particular interest.

\bibliographystyle{IEEEtran}
\bibliography{IEEEabrv,GBbibfile}

\end{document}